\documentclass[11pt]{article}
\pdfoutput=1

\usepackage[margin=1in]{geometry}
\usepackage[T1]{fontenc}
\usepackage{lmodern} 
\usepackage{amsmath,amsfonts,amssymb,amsthm}
\usepackage{mathtools}
\usepackage[usenames,dvipsnames,svgnames,table]{xcolor}
\usepackage{enumitem}
\usepackage{microtype}

\usepackage[pagebackref]{hyperref}
\hypersetup{
    pdftitle={Randomized query complexity can beat certificate complexity}, 
    pdfauthor={Shalev Ben-David and Robin Kothari}, 
    colorlinks=true, 
    linkcolor=blue, 
    citecolor=blue, 
    urlcolor=blue 
}

\renewcommand{\backref}[1]{}

\renewcommand{\backrefalt}[4]{%
\ifcase #1 %
\or
[p.\ #2]%
\else
[pp.\ #2]%
\fi}

\makeatletter
\renewcommand{\paragraph}{%
  \@startsection{paragraph}{4}%
  {\z@}{2ex \@plus 3.3ex \@minus .2ex}{-1em}%
  {\normalfont\normalsize\bfseries}%
}
\makeatother

\newtheorem{theorem}{Theorem}
\newtheorem{lemma}[theorem]{Lemma}

\newcommand{\B}{\{0,1\}}

\DeclareMathOperator{\R}{R}

\DeclareMathOperator{\Q}{Q}
\DeclareMathOperator{\C}{C}
\DeclareMathOperator{\UC}{UC}

\newcommand{\tO}{\widetilde{O}}

\begin{document}
\title{\vspace{-2ex} Randomized query complexity can beat certificate complexity}

\author{
Shalev Ben-David\\
\small University of Waterloo\\
\and
Robin Kothari \\
\small Google Quantum AI\\
}

\date{\vspace{-6ex}}
\maketitle

\begin{abstract}
A long-standing open question in query complexity asks whether there is a total Boolean function $f$ with $\R(f) \ll \C(f)$, where $\R(f)$ and $\C(f)$ denote its bounded-error randomized query complexity and certificate complexity, respectively.
We construct a function with $\R(f)=\tO(\sqrt{\C(f)})$, 
which is optimal up to log factors. The same function also has $\Q(f)=\tO(\C(f)^{1/4})$, where $\Q(f)$ is the bounded-error quantum query complexity of $f$, which is also nearly optimal.
\end{abstract}

\section{Introduction}
\label{sec:intro}

We assume the reader is familiar with standard concepts in query complexity; see \cite{BdW02} for a gentle introduction to the topic. In particular, for a Boolean function $f$, $\C(f)$, $\R(f)$, and $\Q(f)$ denote the certificate complexity, (bounded-error) randomized query complexity and (bounded-error) quantum query complexity respectively. In this paper we resolve the long-standing open question of whether there exists a total function $f$ with $\R(f) \ll \C(f)$.

\begin{theorem}
    For all $n>0$, there is a total Boolean function $f:\B^N\to\B$ with $N=\tO(n^3)$, such that $\C(f)=\Omega(n^2)$, $\R(f)=\tO(n)$, and $\Q(f)=\tO(\sqrt{n})$.
\end{theorem}

Both results are near-optimal because for all total functions $f$, we have $\C(f) = O(\R(f)^2)$ and $\C(f)=O(\Q(f))^4$. Our result can be viewed as a query complexity analogue of a language with low $\textsf{BPP}$ complexity but high $\textsf{NP}$-complexity.

\section{Function and separation}
\label{sec:RvsC}

Consider an international sporting event in which there are $2n$ participating countries and each country sends $n$ athletes. An athlete is labeled $(c,b)$, where $c \in \mathcal C:=[2n]$ denotes their country and $b \in \mathcal B := [n]$ denotes their bib number. Thus there are $|\mathcal C||\mathcal B|=2n^2$ athletes.

Between each pair of distinct countries $c,d\in \mathcal C$ with $c\neq d$, there is a race in which all athletes of the two countries participate. In each of these $\binom{|\mathcal C|}{2}$ races, the valid participants are the athletes of the two countries that label this race and additionally, any non-athletes who wish to participate (e.g., spectators who want to join the race for fun). Let the set of races be denoted by $\mathcal R$.

For every race $\{c,d\}\in \mathcal R$, we have a list $X^{\{c,d\}}$ that stores the race results up to position $n$, with non-athletes represented by the symbol $\bot$. For example, the list $X^{\{c,d\}}$ might look like this:
\begin{equation*}
    (c,7), \ \bot ,\  (d,12), \  (d,2),\ \ldots, \ (c,5).
\end{equation*}
In this race, the athlete with bib number $7$ from country $c$ finished first, a non-athlete finished second, and so on. 

Once the races are over, this sporting event may have a champion or not. An athlete $A$ is the champion if three conditions are met: \\
(C1) $A$ defeated every athlete in all $2n-1$ races in which $A$ participated, \\ 
(C2) $A$ was defeated by strictly less than $n$ non-athletes in total among all $2n-1$ races, and \\
(C3) $A$ submits a valid ``championship proposal form'' that proves conditions (C1) and (C2).

Every athlete $(c,b)\in \mathcal C \times \mathcal B$ can submit a championship proposal form $Y^{(c,b)}$, which informally contains enough information for the judges to quickly verify that conditions (C1) and (C2) are satisfied. Thus this form should list the number of non-athletes that beat athlete $(c,b)$ in each of the $2n-1$ races in which they participated, and these numbers should sum up to $<n$. Instead of writing this list of numbers, we require the athlete to write down the running total (i.e., cumulative sum) of these numbers instead, so that we can quickly check that the last number in this list is $<n$. For example, if in the initial 5 races the athlete's positions were first, second, first, first, and third, simply listing down the number of non-athletes that beat them would produce the list  $(0,1,0,0,2,\ldots)$. Writing this in the cumulative-sum format we get $(0,1,1,1,3,\ldots).$ The second representation is only more powerful than the first since the difference of two adjacent entries reveals the corresponding entry in the first representation.

\paragraph{Formal definition.}
We can now define the function $f$ more formally. The input to $f$ is $(X,Y)$, where both $X$ and $Y$ are bit strings of length $O(n^3 \log n)$. $X$ contains $|\mathcal R|=O(n^2)$ lists $X^{\{c,d\}}$, representing the results of each race, each of which is $O(n \log n)$ bits long. $Y$ contains $|\mathcal C||\mathcal B|$ lists $Y^{(c,b)}$, representing the championship proposal card of each athlete, each of which is $O(n \log n)$ bits. 

Set $f(X,Y)=1$ if and only if there exists an athlete $(c,b)\in \mathcal C \times \mathcal B$ who satisfies these conditions:
\begin{enumerate}[nosep]
    \item Well formatted proposal card: $Y^{(c,b)}$ should be a monotonically non-decreasing list of $2n-1$ numbers such that the last number is $<n$. Thus the proposal card claims a finishing position for athlete $(c,b)$ for each race in which the athlete participates.
    \item Proposal card consistent with race: For each race in which $(c,b)$ participates, the race results variable $X^{\{c,d\}}$ must be consistent with the claimed finishing position in the proposal card and all prior positions of that race must be $\bot$.
\end{enumerate}

Note that for the function to evaluate to $1$ we do not need the entire input to be correctly formatted. For example if athlete $(c,b)$ satisfies the winning condition, it does not matter if athlete $(d,b')$ filled out their proposal card incorrectly. We also do not require that the races in which $(c,b)$ participated have a well-formatted race results list, as long as it satisfies condition 2 above.

As usual, to encode numbers that represent integers up to $n$, we use $\log n$ bits. So, for example, reading the label of the winner of the first race costs $O(\log n)$ queries.

We now establish some properties of this function.

\begin{lemma}\label{lem:C1}
There is at most one champion. Moreover,
$\C_1(f) \leq \UC_1(f) = O(n \log n)$.
\end{lemma}

We note that with the $\C_0(f)$ lower bound, this
reproves a near-quadratic separation between $\UC_1(f)$ and $\C_0(f)$, which was first established by Balodis et al.~\cite{BBGJK23} and recently improved (by removing all log factors) by Pabbaraju~\cite{Pab26}.

\begin{proof}
Two athletes cannot both be the champion, since any pair of athletes share at least one common race (or $2n-1$ common races if they are from the same country), and in this race at most one of these two athletes can be the first athlete to finish the race.

We can check if a given athlete $(c,b)$ is a champion in $O(n \log n)$ queries, as we now describe. This implies $\C_1(f)=O(n \log n)$, where $\C_1(f)$ is the $1$-certificate complexity of $f$. 

The certificate contains the proposal card, the variable $Y^{(c,b)}$, corresponding to athlete $(c,b)$ and the prefix of each race result up to and including $(c,b)$'s claimed position in each race in which athlete $(c,b)$ participates.
Given this certificate, we check that the numbers in the proposal card are monotonically non-decreasing and the last number is $<n$. We then check each claimed position is correct and every participant who defeated the champion in each of those races is a non-athlete. Since there are $<n$ non-athletes, there are only $<n$ such positions, so the length of this entire certificate is $O(n \log n)$. Since this certificate is unambiguous, the claim for $\UC_1(f)$ follows.
\end{proof}

\begin{lemma}
There is a fast randomized algorithm: $\R(f)=O(n \log n)$.
\end{lemma}

\begin{proof}
The randomized algorithm has three stages: we first construct a list of athletes of size $O(n)$ that contains the name of the champion if one exists (with high probability). Then we show how to prune this list to at most $1$ potential champion using $O(n \log n)$ queries. Finally, we read and verify the correctness of this champion's certificate, which is of size $O(n \log n)$, as we established.

For the first stage, note that if there is a champion, this athlete must be in the first position in at least half of their races, since they participate in $2n-1$ races and only $<n$ non-athletes are allowed to finish ahead of the champion in total. 
So if we sample the first position of a constant number of races per country, there is a high chance that we'll see the name of the champion. Since there are $2n$ countries, this only costs $O(n \log n)$ queries and gives us a list of $O(n)$ potential champions.

We now run a tournament between these $O(n)$ potential champions. Given any two potential champions, we can easily eliminate one of them using only $O(\log n)$ queries. To do so, pick any race that is common to both athletes. Both athletes claim certain positions in this race, which we can query. If both claims are correct, we eliminate the athlete who was defeated by the other. We eliminate any athlete that makes an incorrect claim. Thus with $O(\log n)$ queries, we have eliminated at least 1 athlete, and after $O(n \log n)$ queries we are left with 1 candidate champion.

Finally, we read and verify the candidate champion's certificate as described in Lemma \ref{lem:C1}.
\end{proof}

\begin{lemma}
There are no short $0$-certificates: $\C_0(f)=\Omega(n^2)$.
\end{lemma}

\begin{proof}
Consider the input $(X,Y)$ where the $X$ lists for each race are filled with $\bot$ and the $Y$ lists for each athlete are filled with $0$. This is clearly an input with $f(X,Y)=0$. We will show it is hard to certify: Any certificate that exhibits $n^2/2$ bits of this input cannot certify that $f(X,Y)=0$. 

Toward a contradiction, assume there is such a certificate, and it has revealed $n^2/2$ bits of the input. Since there are $2n^2$ athletes, most athletes' championship proposal cards are untouched by the certificate. 
More specifically, if we pick a random athlete, the probability that their proposal card is touched is at most $1/4$.

For a country $c \in \mathcal C$, let $Z_c$ denote the number of bits revealed by the certificate in the
races that involve $c$. We know $\sum_c Z_c \leq n^2$, since one race involves two countries. If we pick a random athlete and look at their country $c$,
we have $\mathbb E[Z_c] \leq n/2$.
By Markov's inequality, $\Pr [Z_c \geq n] \leq 1/2$.

By union bound, the probability of either bad event occurring is at most $3/4$. Thus there exists an athlete $(c,b)$ for which $Z_c < n$ and their entire proposal card is untouched. We can now make this athlete the champion by changing bits that the certificate has not touched. First, since $<n$ bits in their races have been touched, we put this athlete in the first untouched position in each race they participate in. Since all revealed positions are $\bot$, this athlete has been defeated by $<n$ non-athletes in total. Finally, we edit this athlete's proposal card to be consistent with the race positions, which we can do since the proposal card was fully untouched by the certificate.
\end{proof}

\section{Quantum algorithm}
\label{sec:quantum}

\begin{lemma}
    There is an even faster quantum algorithm: $\Q(f)=\tO(\sqrt{n})$.
\end{lemma}

The quantum algorithm follows the same three stages as the randomized algorithm: (1) construct a list of size $O(n)$ that contains the champion if one exists, (2) run a tournament that outputs 1 candidate champion, and (3) verify that this is the champion.
The algorithm uses the quantum algorithm for $\mathtt{SINK}$
(from \cite{SYZ04})
as a subroutine, but otherwise consists only of Grover searches.

\paragraph{First stage: a list of $O(n)$ candidates.}
For the first stage, as we noted before, if we sample the
winner (first position) of a constant number of races per country, there is a high chance that we'll see the name of the champion. If we were to do this explicitly as in the randomized algorithm, this would cost $O(n \log n)$ queries, which is too
much. Instead, we maintain a virtual list of $O(n)$ candidates; they are the candidates whose names appear in the positions the randomized algorithm would have queried. This uses no queries.

\paragraph{Second stage: narrowing down to one candidate.}
For the second stage, we need to run a tournament
between those $O(n)$ potential champions to identify a single
winner, who will then be the single candidate champion of this event. 
We do this by reducing to the 
$\mathtt{SINK}$ problem. In this problem,
the input is an orientation of the complete graph on
$n$ vertices, turning each of the $\binom{n}{2}$ edges
into a directed arc (using $\binom{n}{2}$ input bits);
the goal is to determine whether there is a sink,
meaning a vertex for which all $n-1$ incident arcs are
going into it rather than coming out.

The $\mathtt{SINK}$ problem on $n$ vertices has input size
$\Theta(n^2)$ and can be solved using $\tO(\sqrt{n})$
quantum queries, as shown by
Sun, Yao, and Zhang \cite{SYZ04}.
We will need a quantum algorithm which actually finds
the sink if it exists, rather than just determining
its existence; \cite{SYZ04}
does not explicitly claim that their algorithm finds the sink if one exists, although their algorithm does work. To be fully
rigorous, we can rely on the more recent result of Mande, Paraashar, and Saurabh~\cite{MPS23}, who show that one can always find a ``king'' in such an oriented complete graph
in $\tO(\sqrt{n})$ quantum queries, where a king is a vertex
reachable from every other vertex using at most $2$ hops.
In a general input, a king is not unique; however,
if a sink exists, then the king is unique and equal to the
sink. Thus the king-finding algorithm of \cite{MPS23}
(which also uses $\tO(\sqrt{n})$ quantum queries)
can be used to efficiently find the sink if one exists.

Let's label these potential champions with ID numbers $1$ to $O(n)$.
We use ID numbers in this description because their 
names are unknown to us until we query them at a cost of $O(\log n)$ per name.
We now describe the reduction to $\mathtt{SINK}$: Define a directed graph 
where the vertices are the $O(n)$ potential champions.
For each pair $\{p,q\}$ of such players, we intuitively
say that $p$ beats $q$ (and hence orient the arc so that
it goes from $q$ to $p$) if we can eliminate $q$ as a possible
champion by looking locally at the comparison between
$p$ and $q$, as we did in the randomized algorithm.

More formally, first assume $p$ and $q$ represent distinct athletes.
Pick the first race involving
both $p$ and $q$ (this is unique if $p$ and $q$ represent 
different countries). 
Use the proposal cards of $p$
and $q$ to find their claimed rank in that race,
then look in the results of that race to check whether
these ranks are correct. Finally, if both claimed ranks
are correct, eliminate the player who performed
worse in that race. For example, if the proposal card of $p$
says $p$ ranked 2nd in the race and the proposal card of $q$
claims $q$ ranked 5th, we check positions $2$ and $5$ in the
race results, and if both are correct (listing $p$
and $q$ respectively), we eliminate $q$
(since $2<5$); on the other hand, if $p$ lied on the proposal
card and the race results don't list $p$ in position $2$,
we eliminate $p$ instead. It can happen that both
$p$ and $q$ are eliminated (since both lied); in that case,
we arbitrarily say that $p$ beat $q$ if the ID
of $p$ is smaller than that of $q$. Now we return to the assumption that
$p$ and $q$ were distinct athletes. If they represent the same
athlete, we again break ties using ID number. If one is an athlete 
and the other is $\bot$, we say the athlete wins. Lastly, if both are
$\bot$, we break ties using ID number.

This defines an orientation of the complete graph on
the $O(n)$ candidate champions; moreover,
we can query one bit of this input to $\mathtt{SINK}$
using $O(\log n)$ queries to the real input.
Thus we can run the quantum $\mathtt{SINK}$ algorithm
to get a single candidate champion using
$\tO(\sqrt{n})$ quantum queries.

\paragraph{Third stage: checking the candidate.}
Finally, we verify the correctness of the champion's certificate, as in the randomized algorithm. There are three
things to check. First, we check that the proposal card
of the champion is correctly formatted: it must be monotone
with last entry $<n$. Monotonicity can be checked by searching for a non-monotone consecutive pair using Grover's algorithm, which costs $O(\sqrt{n}\log n)$ quantum queries.

Next, we check that the proposal card correctly lists
the rank of the champion in each relevant race. This can
once again be done via a Grover search: there are $2n-1$
relevant races, and for each one we can verify the correctness
of the proposal card by taking the difference
of two consecutive numbers on the card, going to the 
corresponding entry of the race results $X$, and checking
that the name of the champion is listed there. This can
be done in $\tO(\sqrt{n})$ quantum queries.

Finally, we need to check that in each race, all the racers
who beat the champion are non-athletes. Since these are $<n$ claims
that need to be checked, this can still be
done via a Grover search, but it is slightly more subtle because we
don't know the explicit list of positions we need to check.
We start by learning the total number of positions that
we need to check that are $\bot$; this is simply the last entry
of the proposal card of the candidate champion, say $n'<n$.
We give these $n'$ positions an order: we first order them
by race, then within a race, by finishing position.
For example, if the champion came third in the first race
and fourth in the second race, our order will be:
position $1$ of race $1$,  then position $2$ of race $1$,
then position $1$ of race $2$, then position $2$ of race $2$,
then position $3$ of race $2$ (and so on).

We then Grover
search over the numbers $1,2,\dots,n'$, and for each one,
check if the corresponding position in our order is
$\bot$. For a given $k\le n'$, we can check
whether the $k$-th position is $\bot$ by first binary
searching the proposal card to find the race in which
the $k$-th defeat of the champion occurs
(recall that the proposal card gives the running totals
of the number of players who beat the champion).
For example, if the champion's proposal card lists
$0,0,1,2,7,7,9$
and $k=3$, we can binary search to find $3$
in the list between $2$ and $7$, and conclude that
the third player to defeat the champion is rank $1$
in race $5$. This binary search requires $O(\log^2 n)$
queries for any given $k$, and we can then check the
race outcome is $\bot$ in that position using an additional
$O(\log n)$ queries. There are less than $n$ options for $k$,
so a Grover search can verify all relevant positions are
$\bot$ using $O(\sqrt{n}\log^2 n)$ quantum queries.

\subsection*{Concurrent work}

Yesterday, a paper appeared on the arXiv by Ambainis, Iraids, and Kokainis~\cite{AIK26} establishing a near-quartic separation between $\Q(f)$ and $\C(f)$. 

\subsection*{Human Acknowledgments}

We thank Mika G\"o\"os, Daochen Wang, and Ronald de Wolf for discussions and comments on this paper. Mika G\"o\"os has informed us that he suspected the function constructed in \cite{Pab26} exhibits a nearly quadratic separation between $\R(f)$ and $\C(f)$, and ChatGPT was able to prove this claim.

\subsection*{AI Acknowledgment}

All text in this document is human written, but the result was significantly assisted by AI. We describe the workflow in some detail to highlight some takeaways that might be helpful to others.

First, we assembled a large corpus of useful facts in query complexity over many months using various versions of ChatGPT, Claude, and Gemini. 
Given this curated context and the hunch that there should be a function with $\R(f) \ll \C(f)$ inspired by the Hex function~\cite{BBGJK23}, ChatGPT 5.6 found a function with a power $4/3$ separation. After several rounds asking for improvements, an extremely complicated power 2 separation was discovered. A few rounds of simplification using ChatGPT, Claude, and Gemini brought this down to a \textasciitilde 10 page proof. 

Despite this \textasciitilde$10$ page proof of the quadratic separation between $\C(f)$ and $\R(f)$, the LLMs were unable to construct a power 4 separation between $\C(f)$ and $\Q(f)$. The best separation achieved was power 10/3.  The authors then constructed a new, simpler function with a \textasciitilde$2$ page proof. This function is different from the one in this paper and relied on the existence of bipartite graphs with certain properties guaranteed by a probabilistic argument.
Given this simplified function and the hint that quantum algorithms can find the winner in a tournament quickly, ChatGPT 5.6 was able to find a power 4 separation between $\C(f)$ and $\Q(f)$ that was \textasciitilde$10$ pages. The authors then collaborated with LLMs again to understand and simplify the resulting function to arrive at the simple function presented here that simultaneously achieves both separations.

Our takeaways: (1) Curated context with a proof strategy hint appears to be helpful. We were unable to recover any separation between $\R(f)$ and $\C(f)$ if we gave the LLMs no hints or curated context.
(2) Simplifying AI-generated proofs is useful not just for communicating results to humans but also for LLMs to generalize the results. Without the human-given simplification of the construction,
the LLMs were stuck at a power $10/3$ separation 
between $\Q(f)$ and $\C(f)$.

\bibliographystyle{alphaurl}
\bibliography{refs}

\appendix

\end{document}